\documentclass{article}
\usepackage[utf8]{inputenc}
\usepackage{spconf,amsmath,graphicx}
\usepackage{amsfonts}
\usepackage{amsthm}
\usepackage{hyperref}
\usepackage{amssymb}
\usepackage{cite}
\usepackage{tikz}
\usetikzlibrary{automata, positioning, arrows,calc}
\usepackage{caption}
\usepackage{subcaption}
\usepackage{minibox}

\newcommand\R{\mathbb{R}}

\newcommand\C{\mathbb{C}}
\newcommand\prb{\mathbb{P}}

\newcommand\X{\mathrm{\mathcal{X}}}

\newcommand\rk{\mathrm{rk}}

\newcommand{\argmin}{\operatornamewithlimits{argmin}}

\title{Smoothing complex-valued signals on graphs with Monte-Carlo}

\name{\minibox{\hspace{4cm}Hugo Jaquard\,$^{\dagger\star}$  \quad Michaël Fanuel\,$^{\ddagger\star}$\\ 
		Pierre-Olivier Amblard\,$^{\dagger}$\quad Rémi Bardenet\,$^{\ddagger}$\quad Simon Barthelmé\,$^{\dagger}$\quad Nicolas Tremblay\,$^{\dagger}$\thanks{This work was supported by the ANR grant GRANOLA (ANR-21-CE48-0009), the ERC grant BLACKJACK (ERC-2019-STG-851866), the ANR AI chair BACCARAT (ANR-20-CHIA-0002), as well as the LabEx PERSYVAL-Lab (ANR-11-LABX-0025-01), the Grenoble Data Institute (ANR-15-IDEX- 02), MIAI@Grenoble Alpes (ANR-19-P3IA-0003), the LIA CNRS/Melbourne Univ Geodesic, and the IRS (Initiatives de Recherche Stratégiques) of the IDEX Université Grenoble Alpes.}}}

\address{$^{\dagger}$ GIPSA-lab, CNRS, Univ. Grenoble Alpes, Grenoble INP\\
	$^{\ddagger}$ Univ. Lille, CNRS, Centrale Lille, UMR 9189 – CRIStAL\\
$^\star$ equally contributing authors}

\newtheorem{prop}{Proposition}

\theoremstyle{remark}

\begin{document}

\maketitle

\begin{abstract}
We introduce new smoothing estimators for complex signals on graphs, based on a recently studied \emph{Determinantal Point Process} (DPP). These estimators are built from subsets of edges and nodes drawn according to this DPP, making up trees and \emph{unicycles}, \emph{i.e.}, connected components containing exactly one cycle. We provide a Julia implementation of these estimators and study their performance when applied to a ranking problem.
\end{abstract}

\begin{keywords}
magnetic laplacian, spanning forests, determinantal point processes, graph smoothing, angular synchronization, ranking
\end{keywords}

\section{Introduction}
\label{sect:intro}

\emph{Graph signal processing} (GSP, \cite{shuman2013emerging}) usually considers real data defined over the nodes of a graph $G = (V,E)$, and classically relies on the graph Laplacian. For instance, a typical task in GSP consists in smoothing (denoising) a signal $g \in \R^V$ by solving the penalized (Tikhonov) problem
\begin{equation}
  \label{eq:tikhonov}
  \underset{f \in \mathbb{R}^V}{\argmin} \ q\Vert f - g \Vert^2 + f^\top L f
\end{equation}
where $L\in \R^{V\times V}$ is the graph Laplacian. We will consider throughout weighted and undirected graphs. In Equation~\eqref{eq:tikhonov},
$f^\top L f$ penalises the squared-norm of the discrete derivative on the
graph, \emph{i.e.}:
$$f^\top L f = \sum_{e = (v,v')} w_e ( f(v') - f(v) )^2, $$ where the sum runs
over all edges of the graph and $w_e > 0$ is the edge weight. Thus, the quadratic
form $f^\top L f$ can be thought of as i/~computing the squared difference between
signal values along neighbouring edges, ii/~computing a weighted sum.

For multivariate signals, defining a discrete derivative becomes less obvious, as additional geometry enters the picture. For instance, taking $f(v') - f(v)$ along
an edge $e = (v,v')$ assumes that $f(v')$ and $f(v)$ use the same coordinate system,
and introducing a local change of basis along the edge $(v,v')$ is a way to precise the relation between $f(v)$ and $f(v')$ \cite{singer2012vector}.
We focus on complex-valued signals $f \in \C^V$,
for which a multiplication by  $e^{i
  \theta_{e}}$, where $\theta_{e} \in \left[0 , 2 \pi \right]$, can for instance represent a known phase offset between measurements $f(v')$ and $f(v)$. In this setting, we can define a
\emph{magnetic Laplacian} $L \in \C^{V \times V}$ \cite{kenyon2011spanning} (see Appendix \ref{app:defs}) which acts as follows ($f^*$ denotes the conjugate transpose of $f$):
\begin{equation}
  \label{eq:magnetic-lap}
  f^* L f = \sum_{e = (v,v')} w_e  \vert f(v') - e^{i \theta_{e}} f(v) \vert ^2.
\end{equation}
Note that this equation supposes an orientation of each edge: each $\theta_e$ is thus given with an orientation of $e$. This choice of orientation is however arbitrary: for a given edge $e$, one orientation associated to $\theta_e$ is equivalent to the other orientation associated to $-\theta_e$ as $|f(v') - e^{i \theta_{e}} f(v) | ^2 = |e^{-i \theta_{e}}f(v') -  f(v) | ^2$.
The set $\{e^{i \theta_e}\}_e$ describes a \emph{unitary connection} between the nodes of $G$ \cite{kenyon2011spanning}.

Such Laplacians have applications to synchronization
\cite{singer2011angular} and ranking problems
\cite{stella2009angular,yu2011angular,cucuringu2016sync}, as we explain in Section~\ref{sect:ranking}. 

GSP algorithms that use the graph Laplacian often have $\mathcal{O}(\vert V \vert^3)$ scaling when
implemented exactly, due to the matrix inversions or factorisations that are
used (\emph{e.g.}, the solution of Equation~\eqref{eq:tikhonov} reads $f_o = q(L + qI)^{-1} g$). In large graphs approximate methods are necessary, and
\cite{pilavci2021graph} introduced a Monte-Carlo estimator for the Tikhonov
smoothing problem of Equation~\eqref{eq:tikhonov}, with favourable asymptotic runtime. In this
work we generalise the estimator of \cite{pilavci2021graph} to complex signals, 
using a process recently introduced in \cite{fanuel2022sparsification}.

In Section \ref{sect:MTSF_process}, we introduce a slight variation of the
random process of \cite{fanuel2022sparsification}, defined over graphs with a
unitary connection. Our main theoretical results are given in Section
\ref{SECT:ESTIMATORS}, where we derive estimators for the solution of the
Tikhonov smoothing problem. We describe a practical application to ranking in Section \ref{sect:ranking}. Some technical definitions and proofs are deferred to Appendices \ref{app:defs}, \ref{app:background} and \ref{app:estimators}.

\section{A process over multi-type spanning forests}
\label{sect:MTSF_process}

We are interested in a distribution generalizing both the uniform distribution
over \emph{spanning trees} (UST) of a graph, and the \emph{random spanning
  forests} distribution \cite{avena2013some}. A spanning tree is a subset of
edges $\phi \subseteq E$ such that the graph with nodes $V$ and edges $\phi$ is both connected and without cycles (cycle-free). A rooted
spanning forest (RSF) $\phi \subseteq E \cup V$ is the combination of a spanning
forest and a set
of distinguished nodes called the roots (one root per tree). A variation of these distributions can be defined over the set of spanning forests of \emph{unicycles} (FU) \cite{forman1993determinants,kenyon2011spanning}, subsets of edges containing \emph{exactly one} cycle per connected component, and spanning all the nodes in $V$. The generalization we consider instead draws its samples from the set of rooted \emph{multi-type spanning forests} (MTSF). A rooted MTSF is a spanning subset of edges and nodes $\phi \subseteq E \cup V$ whose connected components are made up of rooted trees and \emph{unicycles}. These different structures are illustrated in Fig. \ref{fig:structures}.

We will use the following distribution, introduced in \cite{fanuel2022sparsification}, over rooted MTSFs $\phi = \phi_\bullet \cup \rho(\phi)$ of $G$:
\begin{equation}
\label{eq:prob_mtsf}
\prb(\phi) \propto \prod_{r \in \rho(\phi)} q_r \prod_{e \in \phi_\bullet} w_e \prod_C \left( 2 - 2\cos(\theta_C) \right),\end{equation}
with $\rho(\phi)$ the set of roots of $\phi$, $\phi_\bullet$ its set of edges and $q_v \in \R^*_+$ positive parameters associated to each node $v$. The third product is over the cycles of the unicycles in $\phi$, where $\theta_C = \sum_{e \in C} \theta_e$.

\begin{figure}[t]
\begin{minipage}[b]{0.48\linewidth}
  \centering
  \centerline{\begin{tikzpicture}[every node/.style={circle,inner sep=1pt,thick,draw,scale=0.76},
	                        every edge/.style={very thick,draw}]
	        \node (1) at (0,0){};
            \node (2) [right of=1]{};
            \node (3) [right of=2]{};
            \node (4) [right of=3]{};
            \node (5) [below of=1]{};
            \node (6) [right of=5]{};
            \node (7) [right of=6]{};
            \node (8) [right of=7]{};
            \node (9) [below of=5]{};
            \node (10) [right of=9]{};
            \node (11) [right of=10]{};
            \node (12) [right of=11]{};
            \node (13) [below of=9]{};
            \node (14) [right of=13]{};
            \node (15) [right of=14]{};
            \node (16) [right of=15]{};
            \path [-,color=red] (1) edge (2);
            \path [-,color=black] (2) edge[dotted] (3);
            \path [-,color=red] (3) edge (4);
            \path [-,color=red] (5) edge (6);
            \path [-,color=black] (6) edge[dotted] (7);
            \path [-,color=black] (7) edge[dotted] (8);
            \path [-,color=black] (9) edge[dotted] (10);
            \path [-,color=black] (10) edge[dotted] (11);
            \path [-,color=red] (11) edge (12);
            \path [-,color=red] (13) edge (14);
            \path [-,color=red] (14) edge (15);
            \path [-,color=black] (15) edge[dotted] (16);
            \path [-,color=black] (1) edge[dotted] (5);
            \path [-,color=red] (2) edge (6);
            \path [-,color=black] (3) edge[dotted] (7);
            \path [-,color=red] (4) edge (8);
            \path [-,color=red] (5) edge (9);
            \path [-,color=red] (6) edge (10);
            \path [-,color=black] (7) edge[dotted] (11);
            \path [-,color=red] (8) edge (12);
            \path [-,color=black] (9) edge[dotted] (13);
            \path [-,color=red] (10) edge (14);
            \path [-,color=red] (11) edge (15);
            \path [-,color=red] (12) edge (16);
        \end{tikzpicture}}
  \centerline{Spanning tree}\medskip
\end{minipage}
\hfill
\begin{minipage}[b]{.48\linewidth}
  \centering
  \centerline{\begin{tikzpicture}[every node/.style={circle,inner sep=1pt,thick,draw,scale=0.76},
	                        every edge/.style={very thick,draw}]
	        \node (1) at (0,0){};
            \node (2) [right of=1]{};
            \node (3) [right of=2]{};
            \node (4) [right of=3]{};
            \node (5) [below of=1]{};
            \node (6) [right of=5]{};
            \node (7) [right of=6]{};
            \node (8) [right of=7]{};
            \node (9) [below of=5]{};
            \node (10) [right of=9]{};
            \node (11) [right of=10]{};
            \node (12) [right of=11]{};
            \node (13) [below of=9]{};
            \node (14) [right of=13]{};
            \node (15) [right of=14]{};
            \node (16) [right of=15]{};
            \path [-,color=red] (1) edge (2);
            \path [-,color=black] (2) edge[dotted] (3);
            \path [-,color=red] (3) edge (4);
            \path [-,color=red] (5) edge (6);
            \path [-,color=black] (6) edge[dotted] (7);
            \path [-,color=red] (7) edge (8);
            \path [-,color=red] (9) edge (10);
            \path [-,color=black] (10) edge[dotted] (11);
            \path [-,color=red] (11) edge (12);
            \path [-,color=red] (13) edge (14);
            \path [-,color=red] (14) edge (15);
            \path [-,color=black] (15) edge[dotted] (16);
            \path [-,color=red] (1) edge (5);
            \path [-,color=red] (2) edge (6);
            \path [-,color=red] (3) edge (7);
            \path [-,color=red] (4) edge (8);
            \path [-,color=red] (5) edge (9);
            \path [-,color=black] (6) edge[dotted] (10);
            \path [-,color=black] (7) edge[dotted] (11);
            \path [-,color=red] (8) edge (12);
            \path [-,color=black] (9) edge[dotted] (13);
            \path [-,color=black] (10) edge[dotted] (14);
            \path [-,color=red] (11) edge (15);
            \path [-,color=red] (12) edge (16);
        \end{tikzpicture}}
  \centerline{Spanning forest of unicycles}\medskip
\end{minipage}

\begin{minipage}[b]{.48\linewidth}
  \centering
  \centerline{\begin{tikzpicture}[every node/.style={circle,inner sep=1pt,thick,draw,scale=0.76},
	                        every edge/.style={very thick,draw}]
	        \node (1) at (0,0){};
            \node (2) [right of=1,color=red,fill,very thick,inner sep=2pt]{};
            \node (3) [right of=2]{};
            \node (4) [right of=3]{};
            \node (5) [below of=1]{};
            \node (6) [right of=5]{};
            \node (7) [right of=6,color=red,fill,very thick,inner sep=2pt]{};
            \node (8) [right of=7]{};
            \node (9) [below of=5]{};
            \node (10) [right of=9,color=red,fill,very thick,inner sep=2pt]{};
            \node (11) [right of=10]{};
            \node (12) [right of=11]{};
            \node (13) [below of=9]{};
            \node (14) [right of=13]{};
            \node (15) [right of=14]{};
            \node (16) [right of=15]{};
            \path [-,color=red] (1) edge (2);
            \path [-,color=black] (2) edge[dotted] (3);
            \path [-,color=red] (3) edge (4);
            \path [-,color=red] (5) edge (6);
            \path [-,color=black] (6) edge[dotted] (7);
            \path [-,color=red] (7) edge (8);
            \path [-,color=black] (9) edge[dotted] (10);
            \path [-,color=red] (10) edge (11);
            \path [-,color=black] (11) edge[dotted] (12);
            \path [-,color=red] (13) edge (14);
            \path [-,color=black] (14) edge[dotted] (15);
            \path [-,color=black] (15) edge[dotted] (16);
            \path [-,color=black] (1) edge[dotted] (5);
            \path [-,color=red] (2) edge (6);
            \path [-,color=black] (3) edge[dotted] (7);
            \path [-,color=red] (4) edge (8);
            \path [-,color=red] (5) edge (9);
            \path [-,color=black] (6) edge[dotted] (10);
            \path [-,color=black] (7) edge[dotted] (11);
            \path [-,color=red] (8) edge (12);
            \path [-,color=black] (9) edge[dotted] (13);
            \path [-,color=red] (10) edge (14);
            \path [-,color=red] (11) edge (15);
            \path [-,color=red] (12) edge (16);
        \end{tikzpicture}}
  \centerline{Rooted spanning forest}\medskip
\end{minipage}
\hfill
\begin{minipage}[b]{0.48\linewidth}
  \centering
  \centerline{\begin{tikzpicture}[every node/.style={circle,inner sep=1pt,thick,draw,scale=0.76},
	                        every edge/.style={very thick,draw}]
	        \node (1) at (0,0){};
            \node (2) [right of=1]{};
            \node (3) [right of=2]{};
            \node (4) [right of=3]{};
            \node (5) [below of=1]{};
            \node (6) [right of=5]{};
            \node (7) [right of=6]{};
            \node (8) [right of=7,color=red,fill,very thick,inner sep=2pt]{};
            \node (9) [below of=5]{};
            \node (10) [right of=9]{};
            \node (11) [right of=10]{};
            \node (12) [right of=11]{};
            \node (13) [below of=9]{};
            \node (14) [right of=13]{};
            \node (15) [right of=14]{};
            \node (16) [right of=15]{};
            \path [-,color=red] (1) edge (2);
            \path [-,color=black] (2) edge[dotted] (3);
            \path [-,color=red] (3) edge (4);
            \path [-,color=red] (5) edge (6);
            \path [-,color=black] (6) edge[dotted] (7);
            \path [-,color=red] (7) edge (8);
            \path [-,color=black] (9) edge[dotted] (10);
            \path [-,color=red] (10) edge (11);
            \path [-,color=black] (11) edge[dotted] (12);
            \path [-,color=red] (13) edge (14);
            \path [-,color=red] (14) edge (15);
            \path [-,color=black] (15) edge[dotted] (16);
            \path [-,color=red] (1) edge (5);
            \path [-,color=red] (2) edge (6);
            \path [-,color=black] (3) edge[dotted] (7);
            \path [-,color=red] (4) edge (8);
            \path [-,color=red] (5) edge (9);
            \path [-,color=black] (6) edge[dotted] (10);
            \path [-,color=black] (7) edge[dotted] (11);
            \path [-,color=red] (8) edge (12);
            \path [-,color=black] (9) edge[dotted] (13);
            \path [-,color=red] (10) edge (14);
            \path [-,color=red] (11) edge (15);
            \path [-,color=red] (12) edge (16);
        \end{tikzpicture}}
  \centerline{Rooted multi-type spanning forest}\medskip
\end{minipage}

\caption{Different subsets of edges and nodes, in \textcolor{red}{red}.}
\label{fig:structures}
\end{figure}
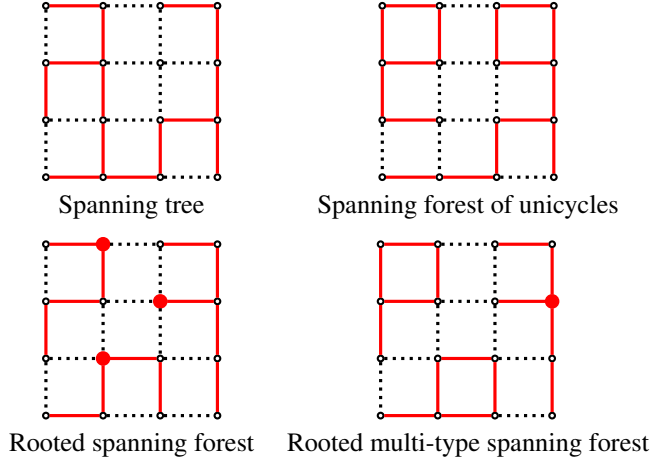

One may wonder why we do not consider the \emph{uniform} distribution over MTSFs and instead favor sampling unicycles for which $2 - 2\cos(\theta_C)$ is large, meaning that these cycles are \emph{inconsistent}. Part of the answer lies in the fact that Eq. \ref{eq:prob_mtsf} is known to describe a \emph{Determinantal Point Process} (DPP) over $E \cup V$, a useful property for proving the results in Section \ref{SECT:ESTIMATORS}. More details are available in Appendix \ref{app:dpps}.

Aforementioned distributions such as USTs and RSFs are conveniently sampled using (a variant of) Wilson's algorithm \cite{wilson1996generating}, based on random walks on the graph. These sampling procedures have been generalized to FUs \cite{kassel2017random} and, recently, to MTSFs \cite{fanuel2022sparsification}, under the sampling condition: \begin{equation}
\label{eq:sampling_condition}
\cos(\theta_\gamma) \geq 0 \text{ for all cycles $\gamma$.}
\end{equation}
Note that this condition applies to \emph{all cycles} $\gamma$ in $G$, and not only in some MTSF $\phi$. The sampling algorithm is recalled in Appendix \ref{app:sampling}.

Bounds on the running time of the sampling algorithm are discussed in \cite{kassel2017random,fanuel2022sparsification}. Here, we only mention that the expected running time is linear in the number of edges, and that it decreases as $\min_v q_v$ increases.

\section{Estimators for smoothing}
\label{SECT:ESTIMATORS}

Monte-Carlo estimators for smoothing on graphs have been developped using RSFs \cite{pilavci2021graph}. The main idea is to sample a rooted RSF $\phi$ before propagating the value of the function at the roots to the other nodes in their associated trees. In the following, we show how this can be generalized to graphs with a unitary connections using the distribution in Equation~\eqref{eq:prob_mtsf}. Proofs are given in Appendix \ref{app:estimators}. Specifically, given $g \in \C^V$ we derive estimators of $f_o = (L + Q)^{-1}Qg$. Here $Q$ is the diagonal matrix of the $q_v$'s  and $L = D - A_\theta$ is the Hermitian \emph{magnetic Laplacian} matrix, with $D$ the diagonal degree matrix and $(A_\theta)_{v',v} = e^{i \theta_{e}}$ if $e = (v,v') \in E$ ($0$ if $e \notin E$). When $q_v = q > 0$ for all $v \in V$, $f_o$ is the optimal solution of the Tikhonov problem:
$$ \underset{f \in \C^V}{\argmin} \ q \Vert f - g \Vert^2 + f^* L f. $$

The first estimator $\tilde{f}$ is built by propagating values through the transport maps $z \mapsto e^{i\theta_{e}} z$ on a rooted multi-type spanning forest $\phi$ sampled according to Equation~\eqref{eq:prob_mtsf}:
$$ \tilde{f}(v) = \left\{ \begin{array}{ll}\psi_{r_\phi(v)\rightarrow v} g(r_\phi(v)) & \text{if the connected component of} \\ & \text{$v$ is a rooted tree,}\\
0 & \text{otherwise,}\end{array}\right.$$
where $r_\phi: V \rightarrow V$ maps nodes to the root of the tree containing them, $a \rightarrow b$ denotes the unique path from $a$ to $b$ in $\phi$ and $\psi_{a \rightarrow b} = \prod_{e \in a \rightarrow b} e^{i \theta_e}$.

\begin{prop}
\label{prop:unbiased}
$\tilde{f}$ is an unbiased estimator of $f_o$: $$\mathbb{E}_\phi(\tilde{f}) = f_o.$$ 
\end{prop}

As a consequence of the Central Limit Theorem, Monte-Carlo estimators converge at
a $\mathcal{O}(\frac{\sigma}{\sqrt{n}})$ rate, with $\sigma$ the standard
deviation of the estimator and $n$ the number of samples. Instead of increasing
the number of samples one may instead focus on decreasing the variance, and in
 our setting this can be done at little additional cost. A first approach is to use a Rao-Blackwell version of the estimator $\tilde{f}$ (see \cite{blackwell1947conditional,rao1992information}) by conditioning on the set of \emph{unrooted} connected components $\pi \subseteq E$ of the MTSF. Consider the estimator:
$$ \overline{f}(v) = \left\{ \begin{array}{ll} \psi_{r_\phi(v)\rightarrow v} h(r_\phi(v)) & \text{if the connected component} \\ & \text{of $v$ is a rooted tree,}\\
0 & \text{otherwise,}\end{array}\right.$$
where $h(u) = \frac{\sum_{w \in T_u} q_w \psi_{w \to u}g(w)}{\sum_{w \in T_u} q_w},$ with $T_v$ the tree containing $v$ in $\phi$. Variance reduction is achieved solely from computing a mean over the nodes of the rooted trees.

\begin{prop}
\label{prop:var_reduction}
We have $\forall v \in V, \overline{f}(v) = \mathbb{E}_{\phi}(\tilde{f}(v) \vert \phi_\bullet = \pi)$ and, as a consequence: $\mathbb{E}_{\pi}(\overline{f}) = \mathbb{E}_\phi(\tilde{f} ) = f_o$. Also, by the law of total variance, $\overline{f}(v)$ has a lower variance than $\tilde{f}(v)$:
$$\mathrm{Var}(\overline{f}(v)) =  \mathrm{Var}(\tilde{f}(v)) - \mathbb{E}(\mathrm{Var}(\tilde{f}(v) \vert \phi_\bullet = \pi))\leq \mathrm{Var}(\tilde{f}(v))$$
\end{prop}

The method of control variates  is another classical variance-reduction technique for Monte-Carlo estimators, adding a term with zero mean to obtain a modified estimator with the same expectation but lower variance \cite{botev2017variance}. The following is proved in \cite{pilavci2021variance}, when $q_v = q$ for all $v \in V$ (\emph{i.e.} $Q = qI$.)

\begin{prop}
Set $\alpha = \frac{2q}{q + 2d_m}$, where $d_m$ is the maximum degree in the graph. Then, the estimator
$$ \hat{f} = \overline{f} - \alpha(q^{-1}(L + qI)\overline{f} - g) $$
is an unbiased estimator of $f_o$ and verifies 
$$\forall v \in V\qquad\mathrm{Var}(\hat{f}(v)) \leq \mathrm{Var}(\overline{f}(v))$$
\end{prop}
For a graph with heterogeneous degree distribution, with \emph{e.g.} a maximum degree much larger than the mean degree, $\hat{f}$ is only a marginal inprovement over $\overline{f}$. However, on graphs with a nearly-homogeneous degree distribution, which is the case for the graphs considered in Section \ref{sect:ranking}, we obtain substantial improvements over $\overline{f}$. 

\section{Ranking from corrupted measurements}
\label{sect:ranking}

\begin{figure*}[t]
\begin{subfigure}[t]{0.20\linewidth}
  \centering
  \centerline{\includegraphics[scale=0.23]{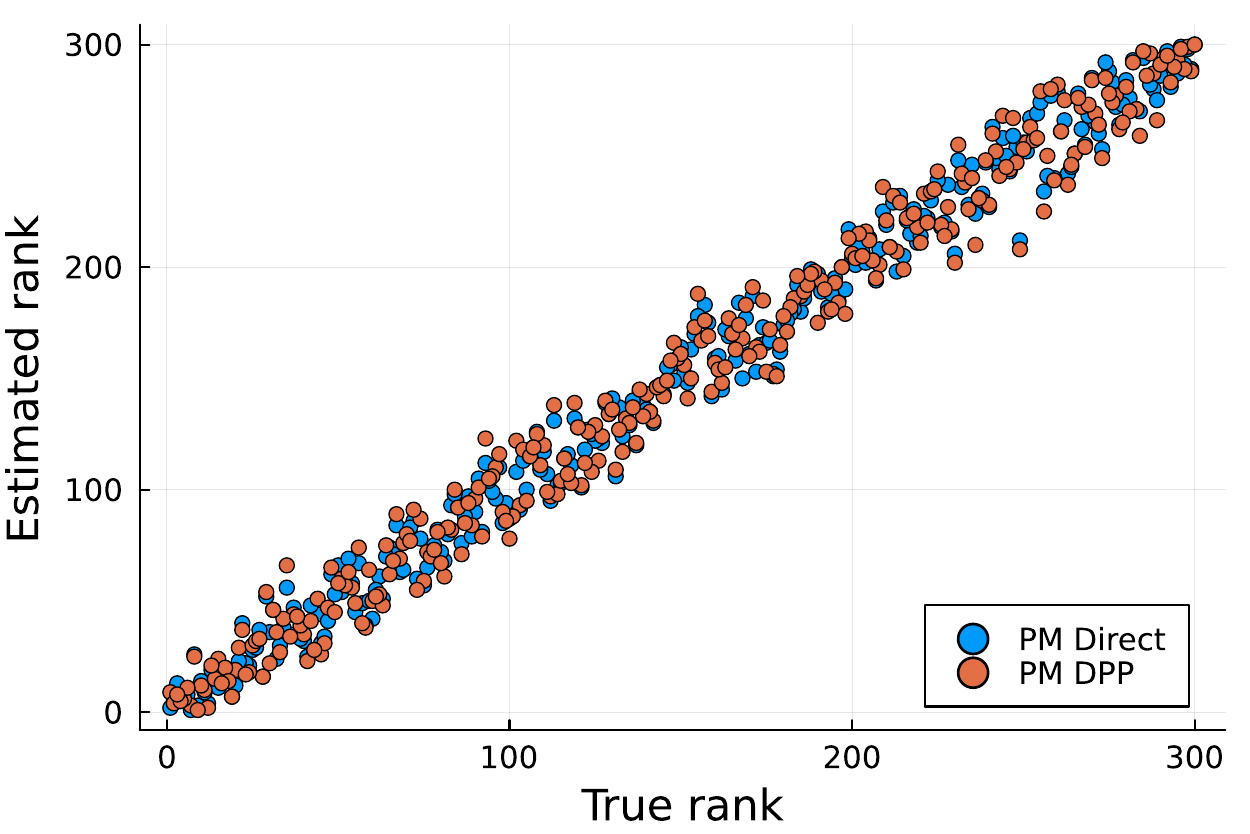}}
  \caption{$p = 0.9, m = 5$}
  \label{subfig:plots08p09}
\end{subfigure}
\hfill
\begin{subfigure}[t]{0.20\linewidth}
  \centering
  \centerline{\includegraphics[scale=0.23]{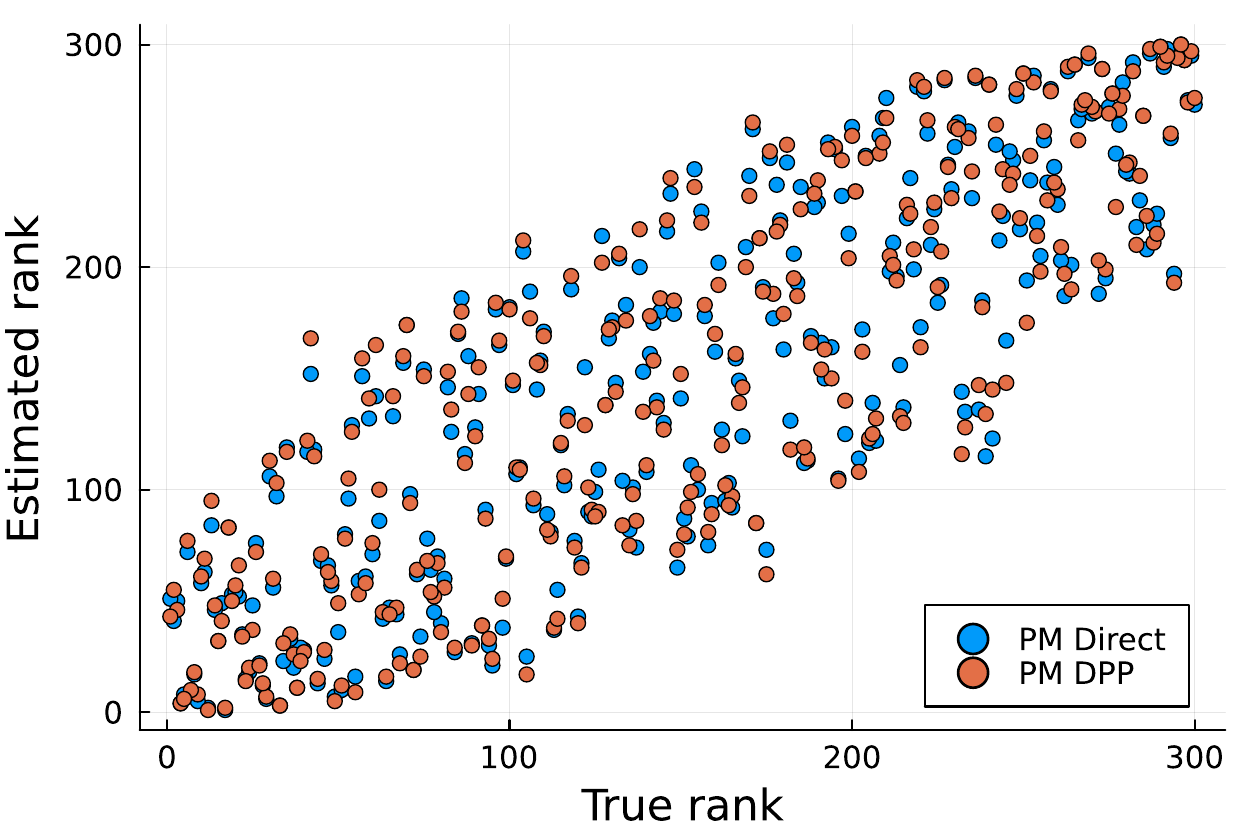}}
  \caption{$p = 0.6, m = 5$}
  \label{subfig:plots08p06}
\end{subfigure}
\hfill
\begin{subfigure}[t]{0.20\linewidth}
  \centering
  \centerline{\includegraphics[scale=0.23]{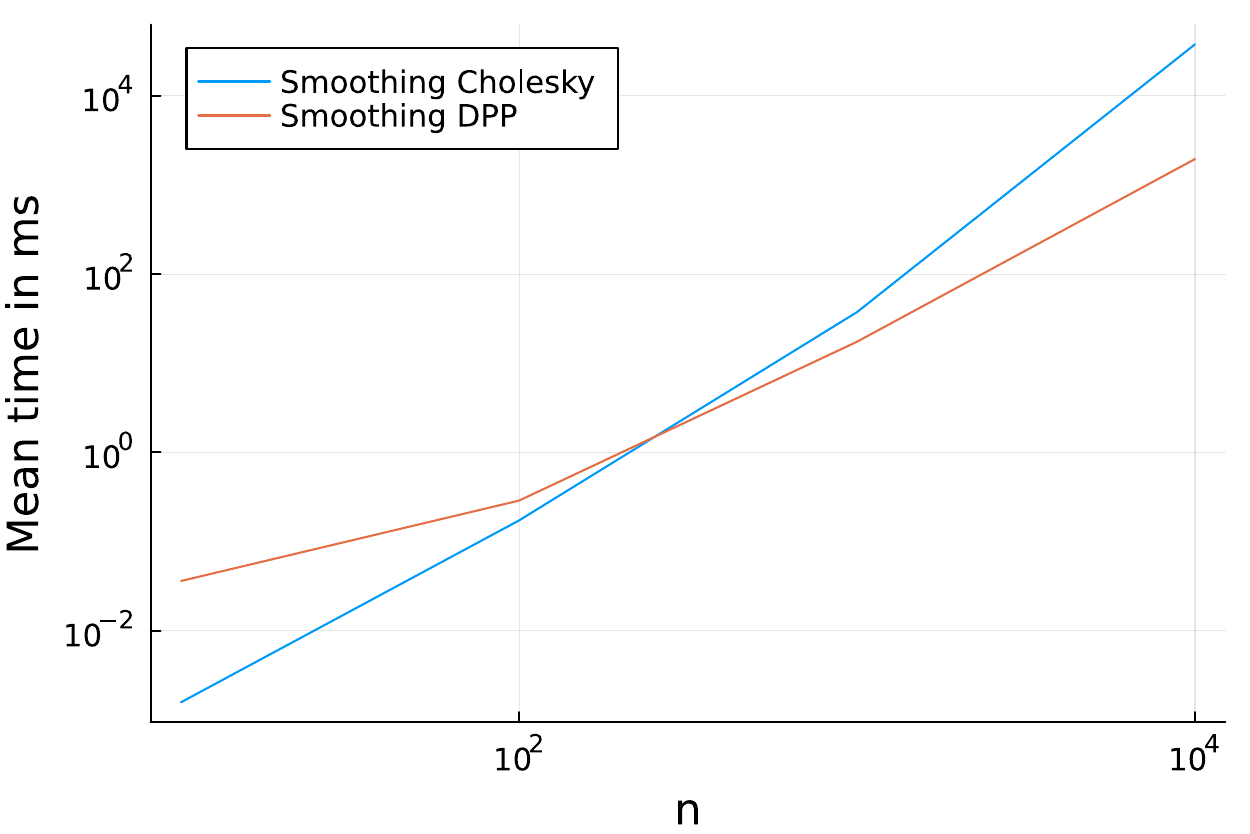}}
  \caption{$p=0.6, m = 5$}
  \label{subfig:smoothing_mean_time}
\end{subfigure}
\hfill
\begin{subfigure}[t]{0.20\linewidth}
  \centering
  \centerline{\includegraphics[scale=0.23]{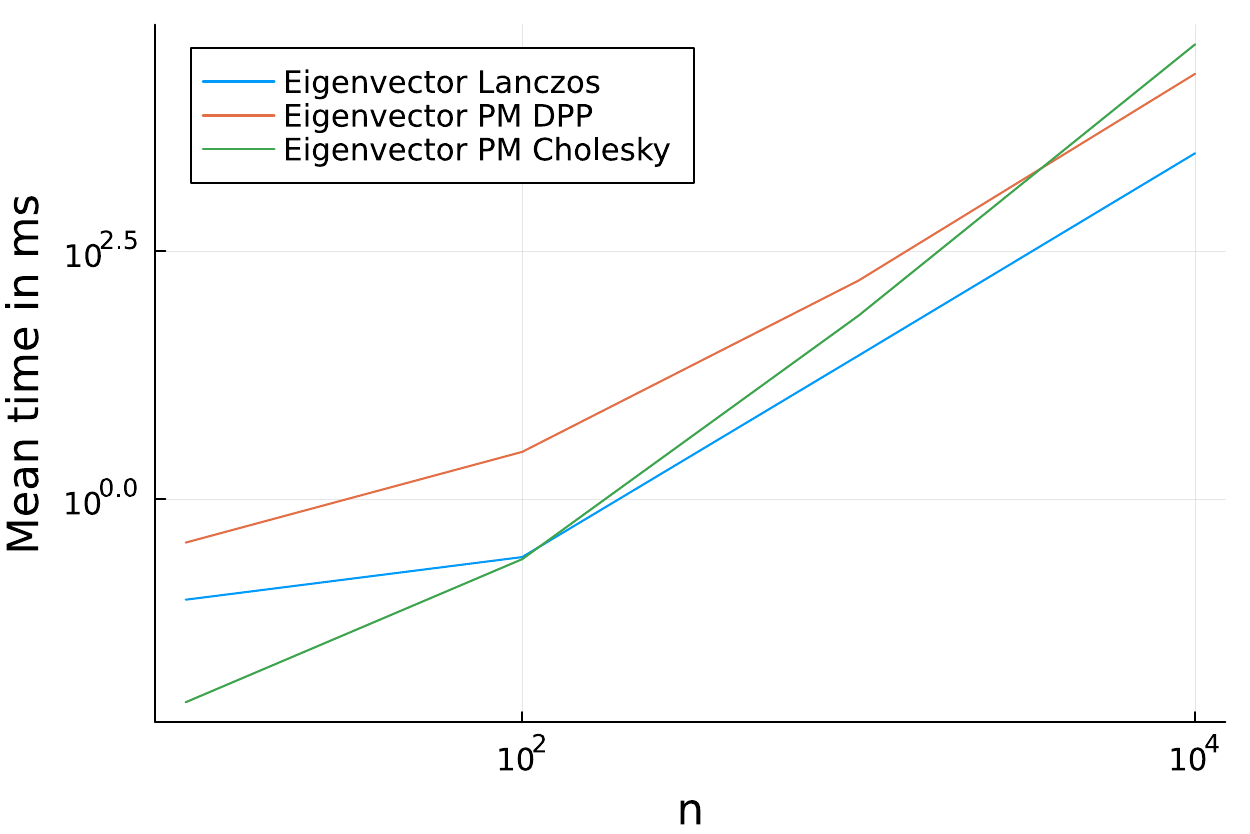}}
  \caption{$p=0.6, m = 5$}
  \label{subfig:vp_mean_time}
\end{subfigure}

\begin{subfigure}[t]{0.20\linewidth}
  \centering
  \centerline{\includegraphics[scale=0.23]{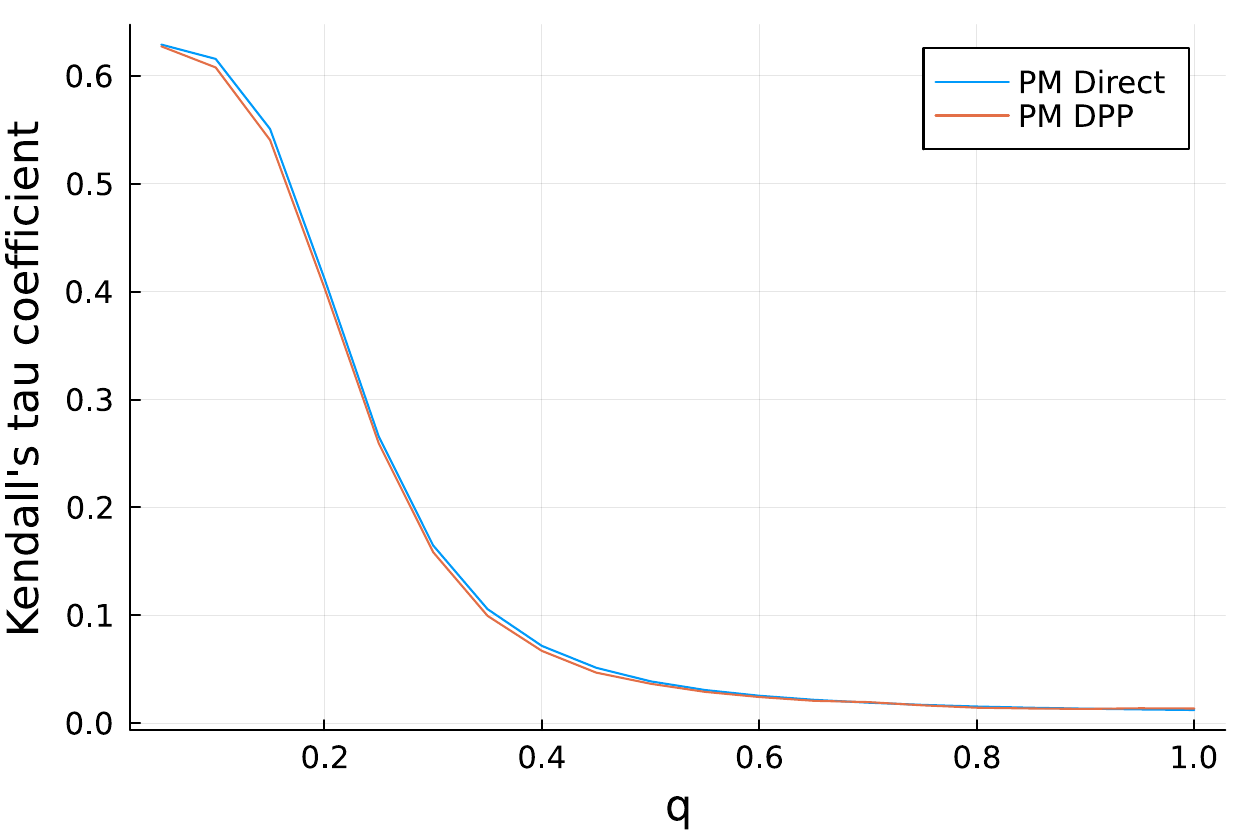}}
  \caption{$k=5,p=0.6, m = 5$}
  \label{subfig:q_kendallk5}
\end{subfigure}
\hfill
\begin{subfigure}[t]{0.20\linewidth}
  \centering
  \centerline{\includegraphics[scale=0.23]{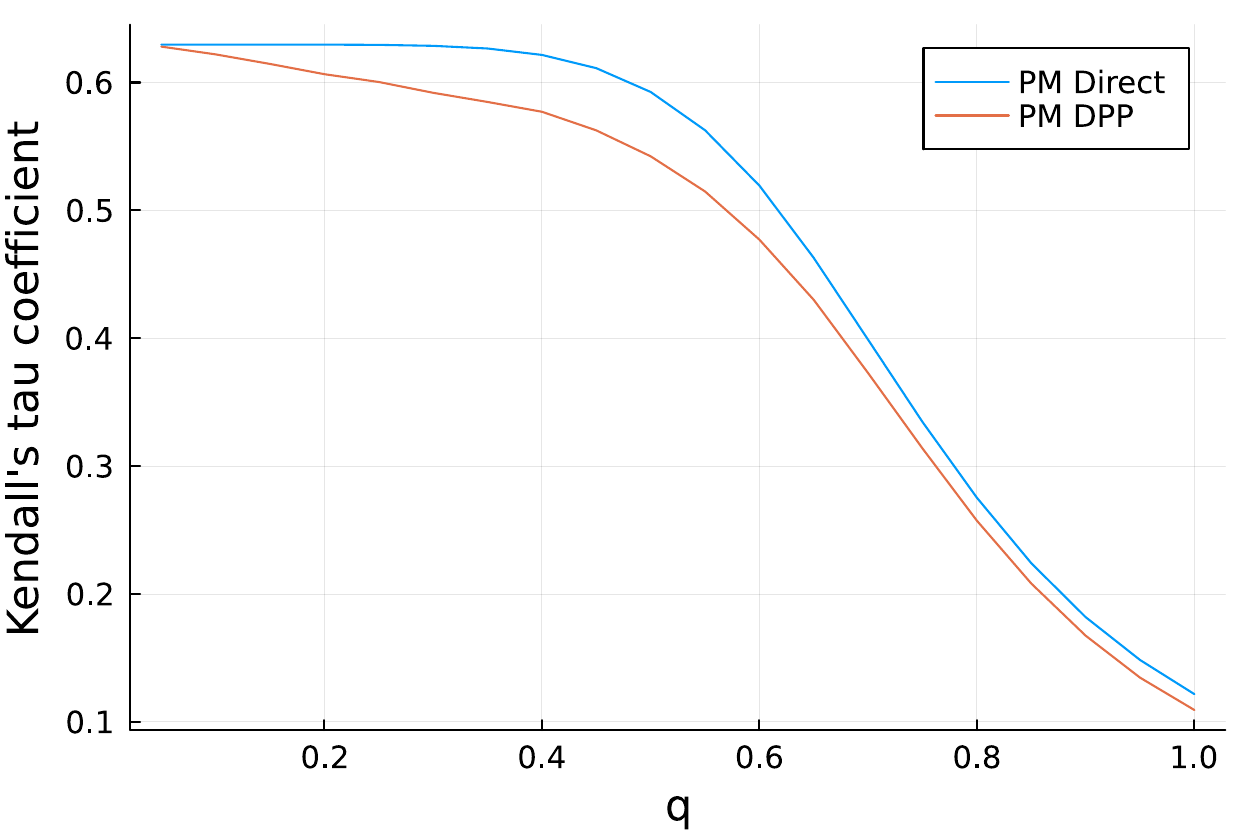}}
  \caption{$k=10,p=0.6, m = 5$}
  \label{subfig:q_kendallk10}
\end{subfigure}
\hfill
\begin{subfigure}[t]{0.20\linewidth}
  \centering
  \centerline{\includegraphics[scale=0.23]{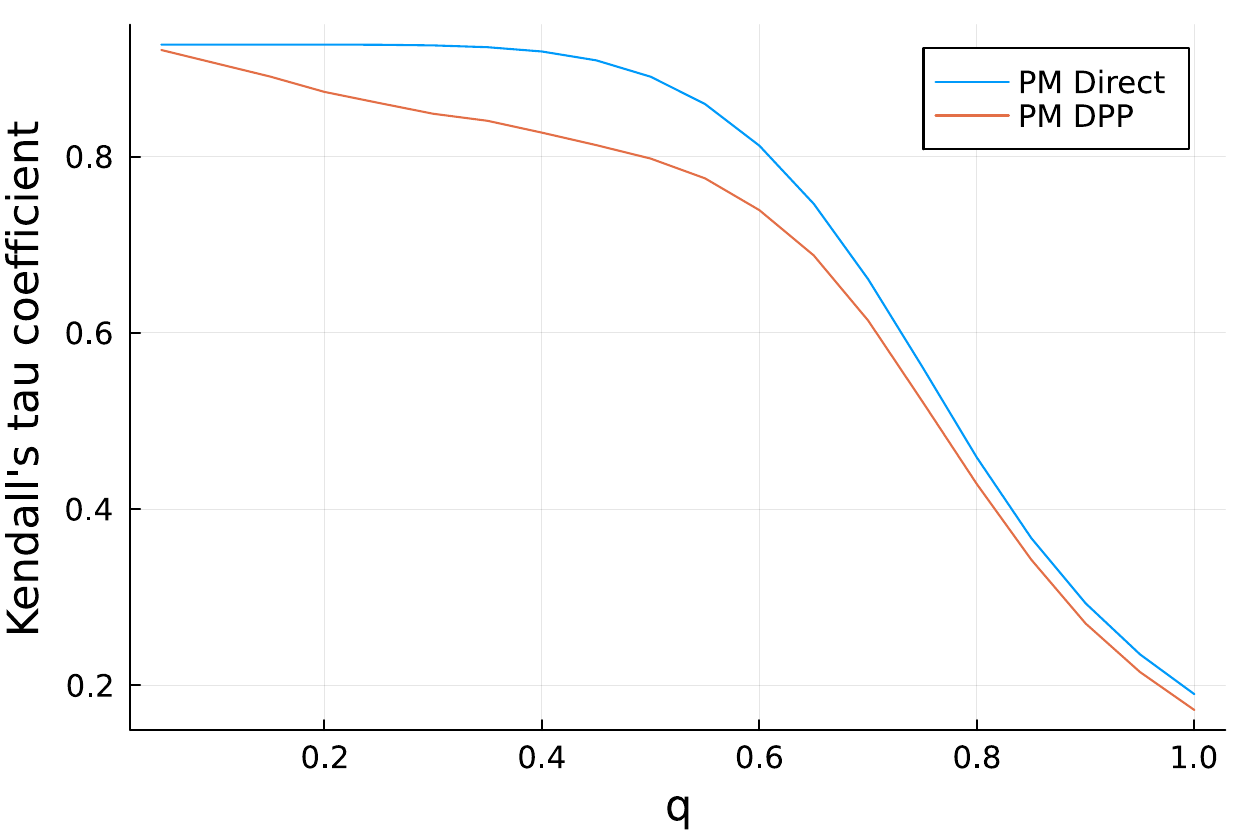}}
  \caption{$k=10,p=0.9, m = 5$}
  \label{subfig:q_kendallk10p09}
\end{subfigure}
\hfill
\begin{subfigure}[t]{0.20\linewidth}
  \centering
  \centerline{\includegraphics[scale=0.23]{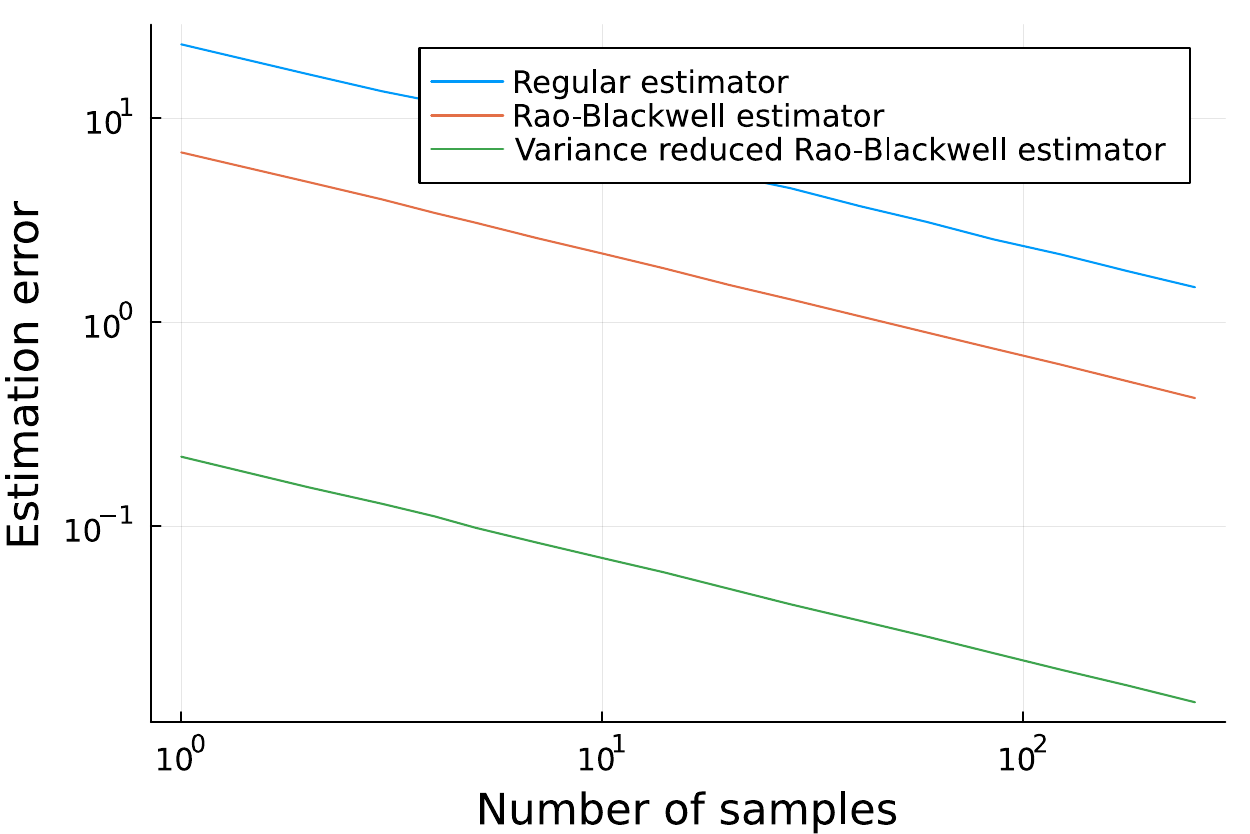}}
  \caption{$p=0.6, m = 5$}
  \label{subfig:benchmarkn3000p06s08}
\end{subfigure}
\caption{Experimental results.}
\label{fig:application}
\end{figure*}

In order to study the actual performance of our estimators, we describe an application to the \emph{ranking} problem. We focus on the ordinal ranking problem, which asks to linearly order a set $X$ of $n$ elements according to an incomplete, possibly incoherent, set of pairwise ordinal comparisons $C_{i,j} \in \{-1,1\}$, for $j > i$.
From this data, build a graph $G$ with $n$ nodes, and a directed edge from $i$ to $j$ if $C_{i,j} = 1$ (resp. from $j$ to $i$ if $C_{i,j} = -1$). Ranking according to \cite{yu2011angular,cucuringu2016sync} then suggests to define the unitary connection $\theta_{i,j} = \frac{\pi \delta C_{i,j}}{n}$ with $\delta \in (0,1)$, and to perform \emph{angular synchronization} \cite{singer2011angular,bandeira2015ten} by solving:
\begin{equation}
\label{eq:synchro}
	\argmin_{(\omega_v)_{v \in V}} \sum_{e = (v,v')} \vert  e^{i \omega_{v'}}- e^{i \theta_e} e^{i \omega_{v}} \vert^2.
\end{equation}
The optimal arguments $\omega_i$ then describe an embedding of the $n$ points onto the unit circle, from which we can extract a ranking (see \cite{cucuringu2016sync} for more details). 

In practice, solving such a non-convex optimization problem is difficult, and a spectral relaxation is considered instead:
$$ \min_{f \in \C^n, \Vert f \Vert^2 = n} f^* \tilde{L} f,$$
for $\tilde{L} = D^{-1/2} L D^{-1/2}$ the normalized graph Laplacian, with $D$ the diagonal degree matrix.
The solution of this classical problem is the eigenvector associated to the smallest eigenvalue of $\tilde{L}$, which can for instance be computed by iterating the map $x \mapsto Mx/\Vert Mx \Vert$ for $M = q(\tilde{L} + qI)^{-1}$, this is the power method \cite{saad2011numerical}. Computing $x \mapsto Mx$ can either be performed directly by solving a linear system, or using the estimators $\tilde{f}$, $\overline{f}$ and $\hat{f}$. 

If we set $\delta = \frac{1}{4}$, the sampling condition in Equation~\eqref{eq:sampling_condition} is satisfied and fast sampling can be achieved using the variation of Wilson's algorithm recalled in Appendix \ref{app:sampling}.

\subsection{Experimental results}

We illustrate in Fig. \ref{fig:application} some results regarding the performance of this approach obtained with our Julia implementation\footnote{\url{https://gricad-gitlab.univ-grenoble-alpes.fr/tremblan/mtsf_for_graph_smoothing}}.
We work with comparisons $C_{i,j}$ randomly obtained from a ground-truth ranking $r_{GT}$ according to the Erdös-Rényi Outliers model \cite{cucuringu2016sync}: a comparison is observed with probability $s$; if a comparison is observed, it follows the ranking $r_{GT}$ with probability $p$, or is chosen uniformly in $\{-1,1\}$ otherwise. We focus on the computationally challenging dense case and set $s = 0.8$. Unless otherwise specified, we use $q=0.1$. 
.
The performance of the estimator $\hat{f}$ in recovering the underlying ranking, using $m=5$ MTSFs \emph{for each} of the $k=10$ iterations of the power method, is illustrated in Figs. \ref{subfig:plots08p09} and \ref{subfig:plots08p06} for two graphs of size $n = 300$ and $p \in \{0.6,0.9\}$. The eigenvector of $L$ computed from the power method \emph{without} using the estimator is also plotted. Here, the initialization vector $y_0$ used for the power method is a random embedding of the $n$ points in the unit circle, spaced out with angle $\frac{\pi}{2n}$.

Runtime benchmarks are available in Figs. \ref{subfig:smoothing_mean_time} and \ref{subfig:vp_mean_time} for the smoothing problem (computation of $My_0$) solved either using $\hat{f}$ or using a Cholesky decomposition, and for the computation of the eigenvector of $M$ for the power method iteration using $\hat{f}$ or a Cholesky decomposition, and for the Lanczos method. We display the mean running time over $100$ measurements on fixed graphs of size $n \in \{10,10^2,10^3,10^4\}$.

We plot Kendall's $\tau$ coefficients \cite{kendall1938new} (the larger the better) for the rankings recovered from the power method in Figs. \ref{subfig:q_kendallk5}, \ref{subfig:q_kendallk10} and \ref{subfig:q_kendallk10p09} for $n=3000$ and varying $q$, averaging over $20$ realisations.

Reconstruction errors $\Vert My - \mathrm{e}(f_o) \Vert$ across varying $m$, averaged over $20$ runs for each of the estimators $\mathrm{e}(f_o)$ of $f_o$, are in Fig. \ref{subfig:benchmarkn3000p06s08} with $n = 3000$. 

\subsection{Discussion}

The results in Figs. \ref{subfig:plots08p09} and \ref{subfig:plots08p06} compare the performance of the reconstructions obtained by the power method with and without using $\hat{f}$, on only one realization of the graph. For these noise regimes, the ground truth ranking $\tau_{GT}$ can no longer be recovered (which would result in a straight diagonal line). In both cases, the recovered ranking is slightly more spread out when using the estimator instead of the exact power method.

We obtain faster runtime in the smoothing problem than a standard Cholesky
decomposition, in graphs with more than $n \sim 10^3$ nodes (Fig.
\ref{subfig:smoothing_mean_time}). For the eigenvector computation in Fig.
\ref{subfig:vp_mean_time}, cross-over occurs at $n \sim 10^4$ against the power method, while we do not outperform a direct Lanczos iteration. Note that performance for the eigenvector computation may be increased by sampling only one set of $m=5$ MTSFs, used for all iterations of the power method. Moreover, our current implementation is far from optimal.

As $q$ decreases, the $\tau$ coefficient increases until it reaches a fixed value in Figs. \ref{subfig:q_kendallk5}, \ref{subfig:q_kendallk10} and \ref{subfig:q_kendallk10p09}. Comparing of Figs. \ref{subfig:q_kendallk5} and \ref{subfig:q_kendallk10} shows that this behavior actually reflects the convergence of the power method, as the spectral gap of $M$ is lower for larger $q$, requiring more iterations to converge. This suggests that there is a trade-off for the best choice of $q$: it should be as large as possible in order to reduce sampling time, and not too large so that it allows fast convergence of the power method.
Note also that Kendall's $\tau$ seems to reach the value of $p$ when convergence occurs. 

Regarding the convergence of the estimators, all three versions of the estimators in Fig. \ref{subfig:benchmarkn3000p06s08} have linear decay in log-log-scale, which is characteristic of the $\mathcal{O}(1/\sqrt{n})$ convergence rate of Monte-Carlo estimators. In our simulations, the estimator $\hat{f}$ improves on the regular estimator by a factor of $10$.

Throughout these experiments, we sampled $m = 5$ MTSFs to compute the estimators. Here, this choice was arbitrary, but we notice nonetheless that this is of the order of the $\mathcal{O}(\log \vert V \vert)$ uniform spanning trees 
necessary to obtain spectral sparsifiers for connection-free graphs~\cite{kaufman2022scalar}, recently adapted to MTSFs in~\cite{fanuel2022sparsification}.

\section{Conclusion}

We define new estimators built by propagating values along edges sampled according to a recently introduced DPP, for the smoothing problem on graphs endowed with a unitary connection, thus generalizing previous existing approaches on graphs. The evaluation of these estimators on the ranking problem, using our Julia implementation, show that they can be advantageous for smoothing starting from moderately sized graphs, as compared to a Cholesky decomposition, while they do not improve on a Lanczos iteration for eigenvector computation. Nonetheless, the proposed estimators exhibit a potentially useful computational property that is uncommon among deterministic algorithms: their computation can be carried out without any prior global knowledge of the graph, from local neighbor queries only, which in some scenarii may be the only practical interaction. The choice of the parameters $q$ and $m$ is application-dependent, and requires further investigation. Possible extensions of this work include designing new applications of the proposed estimators, as well as generalizing this approach to higher-dimensional signals, where current arguments do not carry over. 

\bibliographystyle{IEEEbib}
\bibliography{Template}

\begin{thebibliography}{10}

\bibitem{shuman2013emerging}
David~I Shuman, Sunil~K Narang, Pascal Frossard, Antonio Ortega, and Pierre
  Vandergheynst,
\newblock ``The emerging field of signal processing on graphs: Extending
  high-dimensional data analysis to networks and other irregular domains,''
\newblock {\em IEEE signal processing magazine}, vol. 30, no. 3, pp. 83--98,
  2013.

\bibitem{singer2012vector}
Amit Singer and H-T Wu,
\newblock ``Vector diffusion maps and the connection laplacian,''
\newblock {\em Communications on pure and applied mathematics}, vol. 65, no. 8,
  pp. 1067--1144, 2012.

\bibitem{kenyon2011spanning}
Richard Kenyon,
\newblock ``{Spanning forests and the vector bundle Laplacian},''
\newblock {\em The Annals of Probability}, vol. 39, no. 5, pp. 1983--2017,
  2011.

\bibitem{singer2011angular}
Amit Singer,
\newblock ``Angular synchronization by eigenvectors and semidefinite
  programming,''
\newblock {\em Applied and computational harmonic analysis}, vol. 30, no. 1,
  pp. 20--36, 2011.

\bibitem{stella2009angular}
X~Yu Stella,
\newblock ``Angular embedding: from jarring intensity differences to perceived
  luminance,''
\newblock in {\em 2009 IEEE Conference on Computer Vision and Pattern
  Recognition}. IEEE, 2009, pp. 2302--2309.

\bibitem{yu2011angular}
Stella Yu,
\newblock ``Angular embedding: A robust quadratic criterion,''
\newblock {\em IEEE transactions on pattern analysis and machine intelligence},
  vol. 34, no. 1, pp. 158--173, 2011.

\bibitem{cucuringu2016sync}
Mihai Cucuringu,
\newblock ``Sync-rank: Robust ranking, constrained ranking and rank aggregation
  via eigenvector and sdp synchronization,''
\newblock {\em IEEE Transactions on Network Science and Engineering}, vol. 3,
  no. 1, pp. 58--79, 2016.

\bibitem{pilavci2021graph}
Yusuf~Yi{\u{g}}it Pilavc{\i}, Pierre-Olivier Amblard, Simon Barthelm{\'e}, and
  Nicolas Tremblay,
\newblock ``{Graph Tikhonov regularization and interpolation via random
  spanning forests},''
\newblock {\em IEEE transactions on Signal and Information Processing over
  Networks}, vol. 7, pp. 359--374, 2021.

\bibitem{fanuel2022sparsification}
Micha{\"e}l Fanuel and R{\'e}mi Bardenet,
\newblock ``{Sparsification of the regularized magnetic Laplacian with
  multi-type spanning forests},''
\newblock {\em arXiv preprint arXiv:2208.14797}, 2022.

\bibitem{avena2013some}
Luca Avena and Alexandre Gaudilli{\`e}re,
\newblock ``On some random forests with determinantal roots,''
\newblock {\em Preprint: Weierstra{\ss}-Institut f{\"u}r Angewandte Analysis
  und Stochastik}, vol. 1881, 2013.

\bibitem{forman1993determinants}
Robin Forman,
\newblock ``{Determinants of Laplacians on graphs},''
\newblock {\em Topology}, vol. 32, no. 1, pp. 35--46, 1993.

\bibitem{wilson1996generating}
David~Bruce Wilson,
\newblock ``Generating random spanning trees more quickly than the cover
  time,''
\newblock in {\em Proceedings of the twenty-eighth annual ACM symposium on
  Theory of computing}, 1996, pp. 296--303.

\bibitem{kassel2017random}
Adrien Kassel and Richard Kenyon,
\newblock ``{Random curves on surfaces induced from the Laplacian
  determinant},''
\newblock {\em The Annals of Probability}, vol. 45, no. 2, pp. 932--964, 2017.

\bibitem{blackwell1947conditional}
David Blackwell,
\newblock ``Conditional expectation and unbiased sequential estimation,''
\newblock {\em The Annals of Mathematical Statistics}, pp. 105--110, 1947.

\bibitem{rao1992information}
C~Radhakrishna Rao,
\newblock ``Information and the accuracy attainable in the estimation of
  statistical parameters,''
\newblock in {\em Breakthroughs in statistics}, pp. 235--247. Springer, 1992.

\bibitem{botev2017variance}
Zdravko Botev and Ad~Ridder,
\newblock ``Variance reduction,''
\newblock {\em Wiley statsRef: Statistics reference online}, pp. 1--6, 2017.

\bibitem{pilavci2021variance}
Yusuf Pilavc{\i}, Pierre-Olivier Amblard, Simon Barthelm{\'e}, and Nicolas
  Tremblay,
\newblock ``Variance reduction in stochastic methods for large-scale
  regularised least-squares problems,''
\newblock {\em arXiv preprint arXiv:2110.07894}, 2021.

\bibitem{bandeira2015ten}
Afonso~S Bandeira,
\newblock ``Ten lectures and forty-two open problems in the mathematics of data
  science,''
\newblock {\em Lecture Notes}, 2015.

\bibitem{saad2011numerical}
Yousef Saad,
\newblock {\em Numerical methods for large eigenvalue problems: revised
  edition},
\newblock SIAM, 2011.

\bibitem{kendall1938new}
Maurice~G Kendall,
\newblock ``A new measure of rank correlation,''
\newblock {\em Biometrika}, vol. 30, no. 1/2, pp. 81--93, 1938.

\bibitem{kaufman2022scalar}
Tali Kaufman, Rasmus Kyng, and Federico Sold{\'a},
\newblock ``Scalar and matrix chernoff bounds from $l_\infty$-independence,''
\newblock in {\em Proceedings of the 2022 Annual ACM-SIAM Symposium on Discrete
  Algorithms (SODA)}. SIAM, 2022, pp. 3732--3753.

\bibitem{chung1997spectral}
Fan~RK Chung and Fan~Chung Graham,
\newblock {\em Spectral graph theory}, vol.~92,
\newblock American Mathematical Soc., 1997.

\bibitem{macchi1975coincidence}
Odile Macchi,
\newblock ``The coincidence approach to stochastic point processes,''
\newblock {\em Advances in Applied Probability}, pp. 83--122, 1975.

\bibitem{kulesza2012determinantal}
Alex Kulesza and Ben Taskar,
\newblock ``Determinantal point processes for machine learning,''
\newblock {\em Foundations and Trends{\textregistered} in Machine Learning},
  vol. 5, no. 2--3, pp. 123--286, 2012.

\bibitem{derezinski2021determinantal}
Micha{\l} Derezinski and Michael~W Mahoney,
\newblock ``Determinantal point processes in randomized numerical linear
  algebra,''
\newblock {\em Notices of the American Mathematical Society}, vol. 68, no. 1,
  pp. 34--45, 2021.

\end{thebibliography}

\newpage

\phantom{new page}

\newpage

\appendix

\section{Technical definitions}
\label{app:defs}

We provide some technical definitions used in Appendices \ref{app:background} and \ref{app:estimators}.

Recall the definition of the graph Laplacian $L = \nabla^* \nabla$ \cite{chung1997spectral}, with $\nabla \in \R^{E \times V}$ the (weighted) $\vert E \vert \times \vert V \vert$ edge-vertex incidence matrix.
The graph Laplacian reveals many properties of functions  $f \in \R^V$ defined on $G$. For instance, if the graph is connected, the one-dimensional kernel of $L$ is generated by the constant function. 

When considering complex signals $f \in \C^V$ on $G$, one first defines a \emph{unitary connection} given by a family $\psi_{v,e} = e^{i \theta_{v,e}}$ of unitary complex numbers, extended by $\psi_{v,e} = \psi_{e,v}^*$ and $\psi_{v,v'} = \psi_{e,v'} \psi_{v,e} = e^{i \theta{v,v'}}$. This connection gives rise to a \emph{magnetic Laplacian} $L = \nabla^* \nabla$ \cite{kenyon2011spanning}, with $\nabla$ a twisted discrete differential:
$$
  \nabla_{e,v} = 
  \begin{cases}
    -\sqrt{w_e}\psi_{v,e} & \text{ if $v = a$,} \\    
    \sqrt{w_e}\psi_{v,e} & \text{ if $v = b$,} \\
    0 & \text{otherwise},
  \end{cases}
$$
for $e = (a,b)$.
Unlike the usual graph Laplacian, the magnetic Laplacian has a trivial kernel ($\ker{L} = 0$), unless the connection is \emph{trivial} \cite{fanuel2022sparsification}, which is a key difference in the analysis of the DPP described in Appendix \ref{app:dpps}.

\section{DPP and smoothing estimators}
\label{app:background}

\subsection{Background on DPP and the MTSF process}
\label{app:dpps}

A (discrete) Determinantal Point Process (DPP) is a distribution over subsets of a finite set $\X$ parameterized by a symmetric matrix $K \in M_{\vert \X \vert}(\mathbb{C})$ whose eigenvalues all lie in $\left[ 0, 1 \right]$, this matrix is the \emph{marginal kernel} of the process. We say that $X \subseteq \X$ is distributed according to this DPP if $\prb(A \subseteq X) = \det K_{A,A}$ for every $A \subseteq \X$, where $K_{S,T}$ denotes the submatrix matrix of $K$ whose rows (resp. columns) have been restricted to $S$ (resp. $T$). If $K$ is a projection matrix (i.e. its eigenvalues are either $0$ or $1$), additional structural properties are ensured. For instance, a DPP based of $K$ will only give non-zero probability to samples of size $\rk(K)$, the rank of $K$.

DPPs were introduced as \emph{repulsive} processes \cite{macchi1975coincidence}, and have been sought after in machine learning \cite{kulesza2012determinantal} and randomized linear algebra \cite{derezinski2021determinantal}.

\begin{prop}
The probabilities in Equation~\eqref{eq:prob_mtsf} define a (projective) DPP over $E \cup V$, with marginal kernel $$ K_m = \begin{bmatrix} \nabla \\ \sqrt{Q} \end{bmatrix} (L + Q)^{-1} \begin{bmatrix} \nabla^* \ \sqrt{Q} \end{bmatrix},$$
with $Q$ the $\vert V \vert \times \vert V \vert$ diagonal matrix of the $q_v$'s.
\end{prop}
\begin{proof}
See \cite{fanuel2022sparsification}. This can also be proved by a straightforward extension of the computations in \cite{kenyon2011spanning}.
\end{proof}

This property further relates our study to that of USTs, RSFs and FUs, which have also been showed to be DPPs \cite{avena2013some,kenyon2011spanning}.

When working with DPPs, the Cauchy-Binet formula is an important technical tool, and states that for an $m \times m$ matrix $M = AB$, with $A$ $m \times n$, $B$ $n \times m$ and where $n > m$ we have:
$$ \det M = \sum_{T} \det A_{:,T} \det B_{T,:},$$
where the sum is over $T \subseteq \{1,...,n\}$ of size $\vert T \vert = m$.

\subsection{Sampling}
\label{app:sampling}

We briefly recall the sampling algorithm from \cite{fanuel2022sparsification}. Under the condition in Equation~\eqref{eq:sampling_condition}, a MTSF $\phi$ distributed according to Equation~\eqref{eq:prob_mtsf} is built by iterating the following until all nodes belong to $\phi$.
\begin{itemize}
    \item[a)] Start a random walk from any node $v$ not in $\phi$, while keeping track of the path $p$ built from this random walk. The transitions are as follows: with probability proportional to $q$, set the current node as a root. Otherwise, choose a successor $u$ in the neighborhood of the current node $w$ with probability proportional to $w_{(w,u)}$.
    \item[b)] If a node $r$ is designated as a root, add it to the forest $\phi$ together with the path $p$ and stop the random walk.
    \item[c)] If the path $p$ self-intersects at $w \in V$ before reaching a node in $\phi$, and forms a loop $C$: keep this loop in $p$ with probability $1-\cos(\theta_C)$, add $p$ to the forest $\phi$ and stop the random walk. Otherwise, erase this loop and continue the random walk from $w$.
    \item[d)] If $p$ reaches a node in $\phi$, stop the random walk and add the path $p$ to the forest $\phi$.
\end{itemize}
    
Note that the $\theta_C$ \emph{a priori} depends on the orientation of the loop, but $\cos(\theta_C)$ is orientation-agnostic. Such random walk based sampling algorithms are a noteworthy instance of efficient DPP sampling, which usually require a diagonalization of the kernel $K$ and would result here in $\mathcal{O}((\vert E \vert + \vert V \vert)^3)$ sampling time.

\section{Proofs for Section \ref{SECT:ESTIMATORS}}
\label{app:estimators}

We provide straightforward proofs of the Propositions \ref{prop:unbiased} and \ref{prop:var_reduction} in Section \ref{SECT:ESTIMATORS}. Unless otherwise specified, we work with rooted MTSFs sampled according to Equation~\eqref{eq:prob_mtsf}.

\subsection{Proof of Proposition \ref{prop:unbiased}}

The main part of the proof consists in rewriting $(L + Q)^{-1}$ as an expectation (see Equation~\eqref{eq:laplacian_expectation}). We first use Cramer's rule to express the coefficients of $(L + Q)^{-1}$:
$$ (L + Q)^{-1}_{i,j} = (-1)^{i + j} \frac{\det(L + Q)_{-j,-i}}{\det(L + Q)},$$
where $-k = V \setminus \{k\}$.

For the denominator, we obtain the following sum over MTSFs (see e.g. \cite{fanuel2022sparsification})
$$\det(L + Q) = \sum_\phi \prod_{r \in \rho(\phi)} q_r \prod_{e \in \phi_\bullet} w_e \prod_C (2 - 2 \cos(\theta_C)),$$
which is the normalizing constant in Equation~\eqref{eq:prob_mtsf}.
The Cauchy-Binet formula further allows to express $\det(L + Q)_{-j,-i}$ as:
$$\sum_{\phi \subseteq E \cup V, \vert \phi \vert = \vert V \vert -1} \det \begin{bmatrix} \nabla^* \ \sqrt{Q} \end{bmatrix}_{-j,\phi} \det \begin{bmatrix} \nabla \\ \sqrt{Q} \end{bmatrix}_{\phi,-i},$$
which can be re-arranged to obtain:
\begin{equation}
\label{eq:sum_of_dets}
\sum_{\phi \subseteq E \cup V, \vert \phi \vert = \vert V \vert -1} (-1)^{i+j} \det \begin{bmatrix}\begin{bmatrix} \nabla^* \ \sqrt{Q} \end{bmatrix}_{:,\phi} \ \delta_j \end{bmatrix} \det \begin{bmatrix} \begin{bmatrix} \nabla \\ \sqrt{Q} \end{bmatrix}_{\phi,:} \\ \delta_i^* \end{bmatrix},\end{equation}
with the canonical basis vector $\delta_i \in \R^V$.

The product of determinants in Equation~\eqref{eq:sum_of_dets} can be split along the connected components $\phi^k \subseteq E \cup V$ of $\phi$. 
Using similar computations as \cite{kenyon2011spanning}, we find this product is non-zero if 
one of these connected components $\phi^k$ is a tree containing $i$ and $j$ with no root for which the product of determinants is $\psi_{j \rightarrow i} \prod_{e \in \phi^k_\bullet} w_e$ and, if this component does not span the entire graph, the remaining subgraph contains disjoint connected components which have to be taken among the following subgraphs:
\begin{itemize}
\item a unicycle with cycle $C$ for which the product is $(2 - 2\cos(\theta_C)) \prod_{e \in \phi^k_\bullet} w_e$,
\item a rooted tree with root $r$ for which the product is $q_r \prod_{e \in \phi^k_\bullet} w_e$.
\end{itemize}

In the end, when computing the sum, this translates to:
\begin{align}
\label{eq:laplacian_expectation}
	(L + Q)_{i,j}^{-1} & = \frac{1}{q_j} \sum_{\phi} \mathbf{1}_{i \sim j} \psi_{j \rightarrow i} \prb(\phi) \\
	& = \frac{1}{q_j} \mathbb{E}(\mathbf{1}_{i \sim j} \psi_{j \rightarrow i}), \nonumber
\end{align}
where $\mathbf{1}_{i \sim j}$ is the indicator that $i$ belongs to a tree whose root is $j$.

Writing out the formula for the optimal solution $f_o$ finally yields:
\begin{align}
f_o(i) & = \delta_i^* (L + Q)^{-1}Q g \nonumber \\
& = \sum_j q_j (L+Q)^{-1}_{i,j} g(j) \nonumber \\
	& = \sum_j \mathbb{E}(\mathbf{1}_{i \sim j} \psi_{j \rightarrow i} g(j)) \nonumber \\
& = \mathbb{E}(\tilde{f}(i)),
\end{align}
where the second to last equation is obtained from Equation~\eqref{eq:laplacian_expectation}, and the last line follows from the law of total expectation. 

\subsection{Proof of Proposition \ref{prop:var_reduction}}

Suppose that $\phi_\bullet = \pi$, denote by $r_k$ the root of the connected component $\phi_\bullet^k$, and by $V_k$ the set of vertices in $\phi^k$. Then, if $v$ belongs to $\pi^k$ a tree:
\begin{align}
    \mathbb{E}_{\phi^k}(\tilde{f}(v) \ \vert \ \phi_\bullet^k = \pi^k) & = \sum_{r \in V_k} \prb(r_k = r) \tilde{f}(v) \nonumber \\
    & = \sum_{r \in V_k} \frac{q_r}{\sum_{w \in V_k}q_w} \psi_{r \rightarrow v} g(r) \nonumber \\
    & = \overline{f}(v). \nonumber
\end{align}
If $\pi^k$ is a unicycle, then $\mathbb{E}(\tilde{f}(v)) = 0$.

The two consequences follow from the laws of total expectation and variance.

\end{document}